\documentclass{amsart} 

\usepackage{ amssymb }
\usepackage{ amsthm }
\usepackage{tikz}
\usetikzlibrary{arrows}

\def\Repposen{\mathrm{Rep}^k_{\tikz{\useasboundingbox(.12,0) rectangle (.5,.1);\node[right] at (0,.12) {\tiny \it pos\hspace{-.1mm}.};\node[right] at (0,0) {\tiny \it en.};}}}
\def\rrarrow{  \hspace{.05cm}\mbox{\,\put(0,-2){$\rightarrow$}\put(0,2){$\rightarrow$}\hspace{.45cm}}}
\def\cA{\mathcal A}
\def\cI{\mathcal I}
\def\cT{\mathcal T}

\newtheorem*{defn}{Definition}
\newtheorem{thm}{Theorem}
\newtheorem{lem}[thm]{Lemma}

\newtheorem*{corollary*}{Corollary}
\newtheorem*{lem*}{Lemma}
\theoremstyle{remark}
\newtheorem{remark}[thm]{Remark}
\newtheorem*{remark*}{Remark}

\begin{document}

\author{Andr\'e Henriques} %
\title{Conformal nets are factorization algebras} %

\begin{abstract}
We prove that conformal nets of finite index are an instance of the notion of a factorization algebra.
This result is an ingredient in our proof that, for $G=SU(n)$, the Drinfel'd center of the category of positive energy representations of the based loop group is
equivalent to the category of positive energy representations of the free loop group.
\end{abstract}

\maketitle

\section{Introduction}

In this note, we prove that conformal nets of finite index (Definitions 1.1 and 3.1 in \cite{CN1}) form an instance of the notion of a factorization algebra.
Our main result, Theorem \ref{thm: easy thm}, is a key ingredient in the proof, announced in \cite{CS(pt)}, that the category of solitons of a finite index conformal net is a bicommutant category.

Our main theorem is an analog, within the coordinate-free setup of \cite{CN1}, of the additivity property of conformal nets. 
Let $\cA$ be a conformal net on $S^1$, let $I\subset S^1$ be a closed interval, and let $\{I_i\subset I\}$ be collection of closed intervals whose interiors cover that of~$I$.
\emph{Additivity} is the statement that the von Neumann algebras $\cA(I_i)$ then generate a dense subalgebra in $\cA(I)$:
\[
\bigcup \mathring I_i=\mathring I\quad\Rightarrow\quad \bigvee \cA(I_i)=\cA(I).\qquad(\text{additivity})
\]
The additivity property of chiral conformal nets was proven in \cite{FredenhagenJorss}.
If one takes finitely many intervals $I_i$ whose union is $I$, then the corresponding property is called \emph{strong additivity}:
\[
\bigcup I_i=I\quad\Rightarrow\quad \bigvee \cA(I_i)=\cA(I).\qquad(\text{strong additivity})
\]
It is a result of Longo--Xu that chiral conformal nets of finite index satisfy strong additivity \cite[\S5]{LongoXuDichotomy}.

Let now $I$ be an abstract interval, and $\{I_i\subset I\}$ a finite collection of multi-intervals (a multi-interval is a finite disjoint union of intervals)
satisfying $\bigcup I_i\times I_i=I\times I$.
Equivalently, this is the requirement that for every pair of points $p,q\in I$ there exists an element of the cover that contains both $p$ and $q$.
In  Theorem \ref{thm: easy thm}, we prove that for every coordinate free conformal net $\cA$ of finite index, not necessarily chiral\footnote{A coordinate-free conformal net is called \emph{chiral} if the action of the rotation group on its vacuum sector has positive energy.}, we have:
\[
\bigcup\;\! I_i\,{\times}\, I_i=I\,{\times}\, I \quad\Rightarrow\quad \mathrm{colim}\, \cA(I_i)=\cA(I).\qquad(\text{factorization algebra})\hspace{-.5cm}
\]
The colimit which appears in the right hand side, informally denoted $\mathrm{colim}\, \cA(I_i)$, is that of a diagram involving the algebras $\cA(I_i)$ and $\cA(I_i\cap I_j)$.
(The colimit is defined by a universal property in the category of von Neumann algebras and normal $*$-homomorphisms.)
That diagram is written out in the left hand side of equation \eqref{eq: colimit}, below.
Here, the equation ``$\mathrm{colim}\, \cA(I_i)=\cA(I)$'' is the statement that the natural inclusions $\cA(I_i)\to \cA(I)$ extend to an isomorphism ${\mathrm{colim}\, \cA(I_i)\to\cA(I)}$.

\begin{remark*}
The category of von Neumann algebras and normal $*$-homomorphisms is cocomplete \cite[Prop.\,5.7]{Kornell} (see also \cite[\S7]{Guichardet}).
\end{remark*}

\section{Factorization algebras}

Let $\mathsf{Man}^n$ be the category whose objects are $n$-dimensional manifolds and whose morphisms are embeddings.
We equip it with the symmetric monoidal structure given by disjoint union.
A collection of open subsets $\{U_i\subset M\}$ of a manifold $M$ is a \emph{Weiss cover} if for every finite subset $S\subset M$, there exists an index $i$
such that $S\subset U_i$ \cite[Chapt.\,6]{CostelloGwilliam}.
Equivalently, being a Weiss cover means that for every $n\in \mathbb N$, the condition $\bigcup U_i^n=M^n$ is satisfied.
Let $\mathsf C$ be a symmetric monoidal category.

\begin{defn}[{\cite[Chapt.\,6]{CostelloGwilliam}}]
An $n$-dimensional $\mathsf C$-valued factorization algebra is a symmetric monoidal functor $\cA:\mathsf{Man}^n\to\mathsf C$
which is a co-sheaf with respect to Weiss covers.
\end{defn}

\noindent Here, being a co-sheaf with respect to Weiss covers means that, for every Weiss cover $\{U_i\;\!{\subset}\;\! M\}$, the natural~map
\begin{equation}\label{eq: colimit}
\quad\mathrm{colim}\left(\tikz[yscale=.5, xscale=.9, baseline=32]{\node[scale=.9](a) at (0,5) {$\cA(U_1\cap U_2)$};\node[scale=.9](b) at (0,4) {$\cA(U_1\cap U_3)$};\node[scale=.9](c) at (0,3) {$\cA(U_2\cap U_3)$};\node[scale=.9](d) at (0,2) {$\cA(U_1\cap U_4)$};\node[scale=.9](e) at (0,1) {$\cA(U_2\cap U_4)$};\node[scale=.9] at (0,0) {$\vdots$};\node[scale=.9](A) at (4,4.5) {$\cA(U_1)$};\node[scale=.9](B) at (4,3.5) {$\cA(U_2)$};\node[scale=.9](C) at (4,2.5) {$\cA(U_3)$};\node[scale=.9](D) at (4,1.5) {$\cA(U_4)$};\node[scale=.9] at (4,.5) {$\vdots$};\draw[->] (a) -- (A);\draw[->] (a) -- (B);\draw[->] (b) -- (A);\draw[->] (b.east)+(0,-.2) -- (C);\draw[->] (c) -- (B);\draw[->] (c) -- (C);\draw[->] (d.east)+(0,.15) -- (A);\draw[->] (d) -- (D);\draw[->] (e.east)+(0,.3) -- (B);\draw[->] (e) -- (D);}
\right)\longrightarrow \cA(M)
\end{equation}
is an isomorphism.
For later notational convenience, we abbreviate the left hand side of \eqref{eq: colimit} as $\mathrm{colim}(\{\cA(U_i\cap U_j)\}\rrarrow\{\cA(U_i)\})$.

In this paper, we are interested in $1$-dimensional factorization algebras (or rather, a small variant of the notion of $1$-dimensional factorization algebra) with values in the category of von Neumann algebras and normal $*$-homomorphisms.

An \emph{interval} is an oriented $1$-manifold diffeomorphic to $[0,1]$.
A \emph{multi-interval} is a finite disjoint union of intervals.
Let $\mathsf{INT}^*$ be the category whose objects are multi-intervals and whose morphisms are orientation preserving embeddings,
and let $\mathsf{INT}\subset \mathsf{INT}^*$ be its full subcategory of intervals.
Let $\mathsf{VN}$ denote the category of von Neumann algebras and normal $*$-homomorphisms, equipped with the symmetric monoidal structure given by spatial tensor product.
By the split property \cite[Def.\,1.1]{CN1}, a conformal net can be viewed as a symmetric monoidal functor $\cA:\mathsf{INT}^*\to\mathsf{VN}$.

We introduce a variant of the notion of Weiss cover that accounts for the fact that morphisms in $\mathsf{INT}^*$ are not open but rather closed inclusions.
Given a topological space $X$, a \emph{Weiss $c$-cover} is
a family of closed subsets $\{V_i\subset X\}$ that satisfies $\bigcup \mathring V_i^n=X^n$ for every $n\in \mathbb N$.
Here, $\mathring V_i$ denotes the \emph{relative} interior of $V_i$ inside $X$, i.e., the largest subset of $V_i$ which is open in $X$.
(For example, the relative interior of $[0,1]$ inside $[0,2]$ is the half-open interval $[0,1[$.)

Throughout this paper, all conformal nets are assumed irreducible, i.e., all the algebras $\cA(I)$ are assumed to be factors
(we work with conformal nets in the sense of \cite[Def.\,1.1]{CN1}).
The following statement expresses the idea that conformal nets are factorization algebras:

\begin{thm}\label{thm: main theorem}
Let $\cA:\mathsf{INT}^*\to\mathsf{VN}$ be a conformal net of finite index.
Then $\cA$ is a co-sheaf with respect to Weiss $c$-covers.
Namely, for every multi-interval $I$, and every Weiss $c$-cover of $I$ by multi-intervals $I_i\subset I$, the natural map
\[
q:\,\mathrm{colim}\big(\{\cA(I_i\cap I_j)\}\rrarrow\{\cA(I_i)\}\big)\,\,\to\,\,\cA(I)
\]
is an isomorphism.
\end{thm}

\begin{remark*}
Here, it is crucial to use covers by \emph{closed} multi-intervals.
For a chiral conformal net $\cA$ on $S^1$,
the functor that sends an open multi-interval $U$ to the algebra $\cA(\bar U)$ is never a factorisation algebra, unless $\cA$ is trivial ($\bar U$ denotes the closure of $U$ in $S^1$),
because the map
\[
A:=\underset {\varepsilon\to 0}{\mathrm{colim}}\, \cA([\varepsilon,1-\varepsilon])\to \cA([0,1])
\]
is not an isomorphism.
Let $I_\varepsilon\subset S^1$ be the image of $[\varepsilon,1-\varepsilon]$ under the exponential map $t\mapsto e^{2\pi i t}:\mathbb R\to S^1$.
The obvious isomorphisms $\cA([\varepsilon,1-\varepsilon])\to \cA(I_\varepsilon)$ followed by the standard actions of $\cA(I_\varepsilon)$ on the vacuum Hilbert space $H_0$
yield an irreducible representation of $A$ on $H_0$.
However, $\cA([0,1])$ is a $\mathit{III}_1$ factor, which admits no irreducible representations \cite[Thm.\,2.13]{GabbianiFrohlich}.
\end{remark*}

We can sharpen the above result a little bit.
Given a compact topological space $X$, a collection $\{V_i\subset X\}_{i\in \cI}$ of closed subsets is called a \emph{2-cover} if there exists a finite subset $\cI'\subset \cI$ such that $\bigcup_{i\in \cI'} V_i^2=X^2$.
Any Weiss $c$-cover is a $2$-cover, and Theorem~$\ref{thm: main theorem}$ is a formal consequence of the following stronger result: 

\addtocounter{thm}{-1}
\begin{thm}\hspace{-.3cm}$'$
Let $\cA$ be a conformal net of finite index,
let $I$ be a multi-interval, and let $\{I_i\subset I\}$ be a $2$-cover by multi-intervals.
Assume furthermore that there exists an element of the $2$-cover that contains $\partial I$.
Then the natural map
\[
q:\,\mathrm{colim}\big(\{\cA(I_i\cap I_j)\}\rrarrow\{\cA(I_i)\}\big)\,\,\to\,\,\cA(I)
\]
is an isomorphism.
\end{thm}

When $I$ is connected, the statement of Theorem $\ref{thm: main theorem}'$ simplifies:

\begin{thm}\label{thm: easy thm}
Let $\cA$ be a conformal net of finite index,
let $I$ be an interval, and let $\{I_i\subset I\}$ be a $2$-cover by multi-intervals.
Then the natural map
\[
q:\,\mathrm{colim}\big(\{\cA(I_i\cap I_j)\}\rrarrow\{\cA(I_i)\}\big)\,\,\to\,\,\cA(I)
\]
is an isomorphism.
\end{thm}

\section{Proofs}

In this section, we present the proofs of the above theorems.
We first prove Theorem \ref{thm: easy thm}.
We then prove Theorem $\ref{thm: main theorem}'$ by a slight variation of the argument.
Theorem \ref{thm: main theorem} is then a formal consequence of Theorem $\ref{thm: main theorem}'$.
We begin with some lemmas.
We first note that, when working with multi-intervals, a $2$-cover induces a cover in the usual sense:

\begin{lem}\label{lem no 1}
Let $I$ be a multi-interval and let $\{I_i\subset I\}_{i\in\cI}$ be a 2-cover by multi-intervals.
Then $\bigcup \mathring I_i=\mathring I$.
\end{lem}
\begin{proof}
By definition, $\bigcup_{i\in\cI'} I_i^2=I^2$ for some finite subset $\cI'\subset \cI$.
Given a point $p\in \mathring I$, pick sequences
$(x_n)$ and $(y_n)$ in $I$ converging to $p$ and satisfying $x_n<p<y_n$.
For every $n$, there exists an index $i\in\cI'$ such that $x_n$ and $y_n$ are both in $I_i$.
The set $\cI'$ being finite, there exists an $I_i$ that contains infinitely many $x_n$'s and $y_n$'s.
Since $I_i$ is a multi-interval, it contains $p$ in its interior.
\end{proof}

The next lemma is technical in nature. It is a generalisation of \cite[Lem.\,1.9]{CN1}. 
Let $\cA$ be a conformal net (not necessarily of finite index) and let $I$ be a multi-interval:

\begin{lem}\label{Lem.1.9 of CN1}
Let $\cI=\{I_i\subset I\}$ be a collection of multi-intervals satisfying $\bigcup \mathring I_i=\mathring I$.
Let $\varphi\in\mathrm{Diff}(I)$ be a diffeomorphism in the connected component of the identity, and let $\hat I:=\varphi(I_0)$ for some $I_0\in\cI$.
Let $H$ be a Hilbert space equipped with actions $\rho_i:\cA(I_i)\to B(H)$ satisfying
\begin{enumerate}
\item 
$\rho_i|_{\cA(I_i\cap I_j)}=\rho_j|_{\cA(I_i\cap I_j)}:\cA(I_i\cap I_j)\to B(H)$.
\item
For every $I_j,I_k\in\cI$ and every intervals $J\subset I_j$, $K\subset I_k$ with disjoint interiors, 
the algebras $\rho_j(\cA(J))$ and $\rho_k(\cA(K))$ commute.
\end{enumerate}
Then the actions $\rho_i|_{\cA(\hat I\cap I_i)}$ of $\cA(\hat I\cap I_i)$ on $H$ extend (uniquely) to an action ${\hat\rho:\cA(\hat I)\to B(H)}$.
\end{lem}

\begin{proof}
We write $\rho_0$ for the action of  $\cA(I_0)$ on $H$.
We may assume without loss of generality that $\varphi$ fixes a neighbourhood of $\partial I$.
Provided that is the case, we can write $\varphi$ as a product of diffeomorphisms $\varphi=\varphi_1\circ\ldots\circ \varphi_n$ with $\mathrm{supp}(\varphi_s)\subset \mathring I_{i_s}$ for some $I_{i_s}\in\cI$.
Let $u_s\in \cA(I_{i_s})$ be unitaries s.t. $\mathrm{Ad}(u_s)=\cA(\varphi_s)$ \cite[Def.\,1.1\,(iv)]{CN1}.
Identifying the elements $u_s$ with their images in $B(H)$, we set
\[
\qquad\hat\rho(a):=u_1\ldots u_n\rho_0\big(\cA(\varphi^{-1})(a)\big)u_n^*\ldots u_1^*.
\]
For every $I_\ell\in\cI$ and every sufficiently small interval $K\subset \hat I \cap I_\ell$, we will show that
\begin{equation}\label{eq:diff trick for nets - bis}
\hat\rho|_{\cA(K)}=\rho_\ell|_{\cA(K)}.
\end{equation}
Here, `sufficiently small' means that the intervals $K_s:=\varphi_s^{-1}(\ldots(\varphi_1^{-1}(K)))$ are contained in $I_{k_s}$ for some $I_{k_s}\in\cI$,
and that for every $s'\le n$ either $K_s\subset I_{i_{s'}}$ or $K_s\cap \mathrm{supp}(\varphi_{s'})=\emptyset$.

For every $s\le n$, we claim that
\begin{equation}\label{eq:diff trick for nets - ter}
\quad\qquad u_1\ldots u_s\rho_{k_s}\big(\cA(\varphi_s^{-1}\circ\ldots\circ \varphi_1^{-1})(a)\big)u_s^*\ldots u_1^* = \rho_\ell(a) \qquad \forall a\in \cA(K).
\end{equation}
Equation \eqref{eq:diff trick for nets - bis} is the special case $s=n$. 
We prove \eqref{eq:diff trick for nets - ter} by induction on $s$.
The base case ($s=0$, $k_0=\ell$) is trivial.
The induction step reduces to the equation
\[
\rho_{k_s}\big(\cA(\varphi_{s}^{-1})(b)\big) = u_{s}^* \rho_{k_{s-1}}(b)u_s,
\]
with $b=\cA(\varphi_{s-1}^{-1}\circ\ldots\circ \varphi_1^{-1})(a)$.
Recall that $b\in\cA(K_{s-1})$, $u_s\in \cA(I_{i_s})$ and that, by assumption, either $K_{s-1}\subset I_{i_s}$ or $K_{s-1}\cap \mathrm{supp}(\varphi_s)=\emptyset$.
In the first case, we have
\[
u_s^* \rho_{k_{s-1}}(b)u_s  = u_s^* \rho_{i_s}(b)u_s = \rho_{i_s}(u_s^* b u_s) = \rho_{i_s}\big(\cA(\varphi_s^{-1})(b)\big) = \rho_{k_s}\big(\cA(\varphi_s^{-1})(b)\big).
\]
In the second case, the elements $\rho_{k_{s-1}}(b)$ and $u_s$ commute:
\[
u_s^* \rho_{k_{s-1}}(b)u_s = \rho_{k_{s-1}}(b) = \rho_{k_s}(b) = \rho_{k_s}\big(\cA(\varphi_s^{-1})(b)\big),
\]
where the last equality holds since $b\in A(K_{s-1})$ and $\varphi_s$ acts trivially on $K_{s-1}$.
This finishes the proof of \eqref{eq:diff trick for nets - ter} and hence of \eqref{eq:diff trick for nets - bis}.
Finally, by strong additivity (which is one of the axioms in \cite[Def.\,1.1]{CN1}), it follows from \eqref{eq:diff trick for nets - bis} that
$\hat\rho(a)=\rho_\ell(a)$ for every $a\in \cA(\hat I \cap I_\ell)$.
\end{proof}

We now establish some assumptions under which the hypotheses of Lemma~\ref{Lem.1.9 of CN1} are satisfied:

\begin{lem}\label{prep. lem}
Let $I$ be a multi-interval, and
let $\cI=\{I_i\subset I\}$ be a $2$-cover.
Let $\rho_i:\cA(I_i)\to B(H)$ be actions satisfying $\rho_i|_{\cA(I_i\cap I_j)}=\rho_j|_{\cA(I_i\cap I_j)}$.
Then, for every $I_j,I_k\in\cI$ and every intervals $J\subset I_j$, $K\subset I_k$ with disjoint interiors, 
we have
\[
[\rho_j(\cA(J)),\rho_k(\cA(K))]=0.
\]
\end{lem}

\begin{proof}
We assume without loss of generality that the 2-cover is finite.
The finite set $S:=\bigcup_{i\in\cI} \partial I_i$ decomposes $J$ and $K$ into a finitely many intervals:
$J=J_1\cup\ldots\cup J_n$ and $K=K_1\cup\ldots\cup K_m$.
 For each pair $J_r$, $K_s$ of above intervals, we will argue that
\begin{equation}\label{slgjbnkjs}
[\rho_j(\cA(J_r)),\rho_k(\cA(K_s))]=0.
\end{equation}
Pick interior points $x\in \mathring J_r$ and $y \in \mathring K_s$.
Since $\cI$ is a $2$-cover, there exists an $i\in\cI$ such that $\{x,y\}\subset I_i$.
It follows that $J_r\cup K_s\subset I_i$.
The actions of $\cA(J_r)$ and $\cA(K_s)$ on $H$ factor through that of $\cA(I_i)$, so
equation \eqref{slgjbnkjs} follows.

Equation \eqref{slgjbnkjs} being true for every pair $J_r$, $K_s$ as above,
by strong additivity, it follows that the algebras 
$\rho_j(\cA(J))=\bigvee_r \rho_j(\cA(J_r))$ and $\rho_k(\cA(K))=\bigvee_s \rho_k(\cA(K_s))$ commute with each other.
\end{proof}

The following lemma contains the main argument of the proof of Theorem \ref{thm: easy thm}.
Let $\cA$ be a conformal net of finite index:

\begin{lem}\label{the lem}
Let $I$ be an interval and let $\cI=\{I_i\subset I\}$ be a $2$-cover.
Let $H$ be a Hilbert space equipped with actions $\rho_i:\cA(I_i)\to B(H)$ satisfying $\rho_i|_{\cA(I_i\cap I_j)}=\rho_j|_{\cA(I_i\cap I_j)}$.
Then those maps extend to an action of $\cA(I)$.
\end{lem}

\begin{proof}
We may assume, without loss of generality, that the 2-cover is closed under taking subsets:
($I_i\in\cI$ and $J\subset I_i$, $J$ a multi-interval) $\Rightarrow$ ($J\in \cI$).

By Lemmas \ref{lem no 1} and \ref{prep. lem}, we are in a situation to apply Lemma \ref{Lem.1.9 of CN1}.
The latter implies that for every interval $J\varsubsetneq I$, the actions of $\cA(I_i\cap J)$ extend (uniquely) to an action of $\cA(J)$.
We may therefore assume without loss of generality that $I=[0,5]$, and that the $2$-cover contains the multi-intervals $[0,2]\cup[3,5]$ and $[1,4]$ as elements.

Recall that $L^2(-)$ is the unit for the operation $\boxtimes$ of Connes fusion.
We have
$H\,\cong\, L^2\cA([1,4])\boxtimes_{\cA([1,4])} H$,
both as $\cA([1,4])$-modules and as $\cA([0,1]\cup[4,5])$-modules.
By \cite[Cor.\,2.9]{CN2}, the vacuum sector $L^2\cA([1,4])$ is isomorphic~to
\[
L^2\cA([2,3])\boxtimes_{\cA([2,3]\cup_{\{2,3\}}[\overline{2,3}])} \big(L^2\cA([1,4])\boxtimes_{\cA([1,2]\cup[3,4])}L^2\cA([1,4])\big)
\]
as an $\cA([1,4])$-$\cA([1,4])$-bimodule (this is where we use the assumption that $\cA$ has finite index).
Here, $[\overline{2,3}]$ denotes the interval $[2,3]$ equipped with the opposite orientation, and the algebra
$\cA([2,3]\cup_{\{2,3\}}[\overline{2,3}])$ associated to the circle $[2,3]\cup_{\{2,3\}}[\overline{2,3}]$ is described in \cite[Prop.\,1.25]{CN2}.
The action of $\cA([2,3]\cup_{\{2,3\}}[\overline{2,3}])$ on $L^2\cA([1,4])\boxtimes_{\cA([1,2]\cup[3,4])}L^2\cA([1,4])$
comes from the left action of $\cA([2,3])\subset \cA([1,4])$ on the second copy of $L^2\cA([1,4])$ and the right action of $\cA([2,3])\subset \cA([1,4])$ on the first copy of $L^2\cA([1,4])$.

Let us abbreviate $\cA([a,b])$ by $\cA_{ab}$, $\cA([a,b]\cup[c,d])$  by $\cA_{ab\cup cd}$, and $\cA([a,b]\cup_{\{a,b\}}[\overline{a,b}])$ by
$\cA_{\tikz{\useasboundingbox (-.16,-.12) rectangle (.16,.13);\node at (0,0) {$\scriptstyle ab$};\node[rotate=90] at (0,.12) {$\scriptstyle )$};}}$.
We have:
\def\boxtimessubcirclealgebra{\boxtimes_{\cA_{\tikz{\useasboundingbox (-.15,-.12) rectangle (.1,.1);\node at (0,0) {$\scriptscriptstyle 23$};\node[rotate=90] at (0,.12) {$\scriptscriptstyle )$};}}}}
\begin{gather*}
H\,\cong\, L^2\cA_{14}\boxtimes_{\cA_{14}} H\quad\,\,\,\text{and}\quad\,\,\,
L^2\cA_{14}\,\cong\,L^2\cA_{23}\boxtimessubcirclealgebra \big(L^2\cA_{14}\boxtimes_{\cA_{12\cup 34}}L^2\cA_{14}\big).
\end{gather*}
Combining those two facts, one gets
\begin{equation}\label{eq: 1st set of eqs}
\begin{split}
H&\cong L^2\cA_{14}\boxtimes_{\cA_{14}}H\\
&\cong 
\big(L^2\cA_{23}\boxtimessubcirclealgebra(L^2\cA_{14}\boxtimes_{\cA_{12\cup 34}}L^2\cA_{14}) \big)\boxtimes_{\cA_{14}}H\\
&\cong\, L^2\cA_{23}\boxtimessubcirclealgebra
\big(L^2\cA_{14}\boxtimes_{\cA_{12\cup 34}}L^2\cA_{14}\boxtimes_{\cA_{14}}H\big) \\
&\cong\, L^2\cA_{23}\boxtimessubcirclealgebra
\big(L^2\cA_{14}\boxtimes_{\cA_{12\cup 34}}H\big).
\end{split}
\end{equation}
Using that $H\cong L^2\cA_{02\;\!\cup\;\! 35}\boxtimes_{\cA_{02\cup 35}}H$
and the existence of a (non-canonical) isomorphism
\[
L^2\cA_{14}\boxtimes_{\cA_{12\cup 34}}L^2\cA_{02\;\!\cup\;\! 35}
\cong
L^2\cA_{02}\boxtimes_{\cA_{12}^{\mathit{op}}}L^2\cA_{14}\boxtimes_{\cA_{34}}L^2\cA_{35}
\cong
L^2\cA_{05}
\]
which is compatible with the left actions of $\cA_{14}$ and $\cA_{01\;\!\cup\;\! 45}$ and the right actions of $\cA_{02\;\!\cup\;\! 35}$ and $\cA_{23}$ (\cite[Cor.\,1.33]{CN1} and \cite[Lem.\,A.4]{CN2}),
we get the following sequence of isomorphisms
of $\cA_{14}$- and $\cA_{01\;\!\cup\;\! 45}$-modules:
\begin{equation}\label{eq: 2nd set of eqs}
\begin{split}
H\,&\cong\, L^2\cA_{23}\boxtimessubcirclealgebra\big(L^2\cA_{14}\boxtimes_{\cA_{12\cup 34}}H\big)\\
&\cong\, L^2\cA_{23}\boxtimessubcirclealgebra\big(L^2\cA_{14}\boxtimes_{\cA_{12\cup 34}}L^2\cA_{02\cup 35}\boxtimes_{\cA_{02\cup 35}}H\big)\\
&\cong\,
L^2\cA_{23}\boxtimessubcirclealgebra\big(L^2\cA_{05}\boxtimes_{\cA_{02\cup 35}}H\big).
\end{split}
\end{equation}
The actions of $\cA_{14}$ and of $\cA_{01\;\!\cup\;\! 45}$~on 
\[
L^2\cA_{23}\boxtimessubcirclealgebra\big(L^2\cA_{05}\boxtimes_{\cA_{02\cup 35}}H\big)
\]
extend to an action of $\cA_{05}$ because they both act on $L^2\cA_{05}$.
The actions of $\cA_{14}$ and of $\cA_{01\;\!\cup\;\! 45}$~on $H$ therefore also extend to an action of $\cA_{05}$.
\end{proof}

To help the reader digest the argument in the above proof,
we include a graphical rendering of the isomorphisms which appear in \eqref{eq: 1st set of eqs} and \eqref{eq: 2nd set of eqs}:
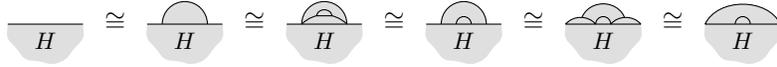
\medskip
\[
\tikz[baseline=-2]{
\fill[gray!25] (.5,0) -- +(180:.5) to[rounded corners=2] +(200:.55) to[rounded corners=2] +(220:.5) to[rounded corners=2] +(240:.55) to[rounded corners=2] +(260:.5) to[rounded corners=2] +(280:.55) to[rounded corners=2] +(300:.55) to[rounded corners=2] +(320:.5) to[rounded corners=2] +(340:.55) to[sharp corners] +(360:.5) -- cycle;
\draw (0,0) --node[below, scale=.9]{$H$} (1,0);
}
\,\,\,\cong\,\,\,
\tikz[baseline=-2]{
\fill[gray!25] (.5,0) -- +(180:.5) to[rounded corners=2] +(200:.55) to[rounded corners=2] +(220:.5) to[rounded corners=2] +(240:.55) to[rounded corners=2] +(260:.5) to[rounded corners=2] +(280:.55) to[rounded corners=2] +(300:.55) to[rounded corners=2] +(320:.5) to[rounded corners=2] +(340:.55) to[sharp corners] +(360:.5) -- cycle;
\draw[fill=gray!25] (.2,0) arc(180:0:.3);
\draw (0,0) --node[below, scale=.9]{$H$} (1,0);
}
\,\,\,\cong\,\,\,
\tikz[baseline=-2]{
\fill[gray!25] (.5,0) -- +(180:.5) to[rounded corners=2] +(200:.55) to[rounded corners=2] +(220:.5) to[rounded corners=2] +(240:.55) to[rounded corners=2] +(260:.5) to[rounded corners=2] +(280:.55) to[rounded corners=2] +(300:.55) to[rounded corners=2] +(320:.5) to[rounded corners=2] +(340:.55) to[sharp corners] +(360:.5) -- cycle;
\draw[fill=gray!25] (.2,0) arc(180:0:.3);
\draw (0,0) --node[below, scale=.9]{$H$} (1,0);
\draw (.2,0) to[in=90+45, out=90-45] (.8,0);
\draw[fill=gray!25] (.4,.111) -- (.6,.111) arc (0:180:.1 and .08);
}
\,\,\,\cong\,\,\,
\tikz[baseline=-2]{
\fill[gray!25] (.5,0) -- +(180:.5) to[rounded corners=2] +(200:.55) to[rounded corners=2] +(220:.5) to[rounded corners=2] +(240:.55) to[rounded corners=2] +(260:.5) to[rounded corners=2] +(280:.55) to[rounded corners=2] +(300:.55) to[rounded corners=2] +(320:.5) to[rounded corners=2] +(340:.55) to[sharp corners] +(360:.5) -- cycle;
\draw[fill=gray!25] (.2,0) arc(180:0:.3);
\draw (0,0) --node[below, scale=.9]{$H$} (1,0);
\draw (.4,0) arc (180:0:.1);
}
\,\,\,\cong\,\,\,
\tikz[baseline=-2]{
\fill[gray!25] (.5,0) -- +(180:.5) to[rounded corners=2] +(200:.55) to[rounded corners=2] +(220:.5) to[rounded corners=2] +(240:.55) to[rounded corners=2] +(260:.5) to[rounded corners=2] +(280:.55) to[rounded corners=2] +(300:.55) to[rounded corners=2] +(320:.5) to[rounded corners=2] +(340:.55) to[sharp corners] +(360:.5) -- cycle;
\draw[fill=gray!25] (.2,0) arc(180:0:.3 and .27);
\path (0,0) --node[below, scale=.9]{$H$} (1,0);
\draw (.4,0) arc (180:0:.1);
\draw[fill=gray!25, line join=bevel] (.6,0) to[in=90+40, out=90-40] (1,0) -- (0,0) to[in=90+40, out=90-40] (.4,0);
}
\,\,\,\cong\,\,\,
\tikz[baseline=-2]{
\fill[gray!25] (.5,0) -- +(180:.5) to[rounded corners=2] +(200:.55) to[rounded corners=2] +(220:.5) to[rounded corners=2] +(240:.55) to[rounded corners=2] +(260:.5) to[rounded corners=2] +(280:.55) to[rounded corners=2] +(300:.55) to[rounded corners=2] +(320:.5) to[rounded corners=2] +(340:.55) to[sharp corners] +(360:.5) -- cycle;
\draw[fill=gray!25, line join=bevel] (0,0) to[in=90+25, out=90-25] (1,0) -- cycle;
\path (0,0) --node[below, scale=.9]{$H$} (1,0);
\draw (.4,0) arc (180:0:.1);
}
\medskip
\]

With Lemma~\ref{the lem} in place, the proof of Theorem \ref{thm: easy thm} is now easy:

\begin{proof}[Proof of Theorem \ref{thm: easy thm}]
We first note that, by the strong additivity property of conformal nets \cite[Def.\,1.1]{CN1}, the map $q$ has dense image.
It is therefore surjective, as any morphism of von Neumann algebras whose image is dense is automatically surjective \cite[Chapt.\,III, Prop.\,3.12]{Tak1}.

To show that $q$ is injective, pick a faithful representation
\[
\pi:\mathrm{colim}\big(\{\cA(I_i\cap I_j)\}\rrarrow\{\cA(I_i)\}\big)\to B(H)
\]
and let $\rho_i:=\pi|_{\cA(I_i)}$.
By Lemma~\ref{the lem}, this extends to an action $\rho:\cA(I)\to B(H)$.
As $\pi$ is injective and $\pi=\rho\circ q$, the map $q$ is also injective.
\end{proof}

The proof of Theorem $\ref{thm: main theorem}'$ follows closely that of Theorem \ref{thm: easy thm}.

\begin{proof}[Proof of Theorem $\mathit{\ref{thm: main theorem}}'$]
Let $I$ be a multi-interval, and let $\cI=\{I_i\subset I\}$ be a $2$-cover one of whose elements contains $\partial I$.
As in the proof of Theorem \ref{thm: easy thm}, it is enough to argue that if $H$ is a Hilbert space equipped with actions $\rho_i:\cA(I_i)\to B(H)$ satisfying $\rho_i|_{\cA(I_i\cap I_j)}=\rho_j|_{\cA(I_i\cap I_j)}$, then
those extend to an action of $\cA(I)$.

Without loss of generality, we may assume that the 2-cover is closed under taking subsets:
($I_i\in\cI$ and $J\subset I_i$) $\Rightarrow$ ($J\in \cI$).
In particular, we may assume that the 2-cover admits an element which
doesn't intersect its boundary, and which has exactly one connected component in each connected component on $I$.

By Lemmas \ref{lem no 1} and \ref{prep. lem}, we are in a situation to apply Lemma \ref{Lem.1.9 of CN1}.
We can therefore assume, without loss of generality, that $I=\bigsqcup_{k=1}^n[0,5]$, and that $\cI$ contains the multi-intervals $\bigsqcup_{k=1}^n([0,2]\cup[3,5])$ and $\bigsqcup_{k=1}^n[1,4]$ as elements.

Following the structure of the proof of Lemma \ref{the lem}, we have isomorphisms:
\def\boxtimessubcirclealgebra{\boxtimes_{\cA^{\otimes\;\!\! n}_{\tikz{\useasboundingbox (-.15,-.05) rectangle (.1,.13);\node at (0,0) {$\scriptscriptstyle 23$};\node[rotate=90] at (0,.12) {$\scriptscriptstyle )$};}}}}
\[
\begin{split}
H&\textstyle\cong L^2(\cA_{14})^{\otimes n}\boxtimes_{\cA^{\otimes n}_{14}}H\\
&\textstyle\cong 
\big(L^2(\cA_{23})^{\otimes n}\boxtimessubcirclealgebra(L^2(\cA_{14})^{\otimes n}\boxtimes_{\cA^{\otimes n}_{12\cup 34}}L^2(\cA_{14})^{\otimes n}) \big)\boxtimes_{\cA^{\otimes n}_{14}}H\\
&\textstyle\cong L^2(\cA_{23})^{\otimes n}\boxtimessubcirclealgebra
\big(L^2(\cA_{14})^{\otimes n}\boxtimes_{\cA^{\otimes n}_{12\cup 34}}L^2(\cA_{14})^{\otimes n}\boxtimes_{\cA^{\otimes n}_{14}}H\big) \\
&\textstyle\cong L^2(\cA_{23})^{\otimes n}\boxtimessubcirclealgebra
\big(L^2(\cA_{14})^{\otimes n}\boxtimes_{\cA^{\otimes n}_{12\cup 34}}H\big).\\
&\textstyle\cong L^2(\cA_{23})^{\otimes n}\boxtimessubcirclealgebra\big(L^2(\cA_{14})^{\otimes n}\boxtimes_{\cA^{\otimes n}_{12\cup 34}}L^2(\cA_{02\cup 35})^{\otimes n}\boxtimes_{\cA^{\otimes n}_{02\cup 35}}H\big)\\
&\textstyle\cong
L^2(\cA_{23})^{\otimes n}\boxtimessubcirclealgebra\big(L^2(\cA_{05})^{\otimes n}\boxtimes_{\cA^{\otimes n}_{02\cup 35}}H\big)
\end{split}
\]
of $\cA(\bigsqcup^n([0,2]\cup[3,5]))$- and $\cA(\bigsqcup^n[1,4])$-modules.\smallskip

The actions of
$\cA(\bigsqcup^n([0,2]\cup[3,5]))$ and \smallskip of $\cA(\bigsqcup^n[1,4])$
on the Hilbert space
$\textstyle L^2(\bigotimes^n\!\!\cA_{23})\boxtimessubcirclealgebra\big(L^2(\bigotimes^n\!\!\cA_{05})\boxtimes_{\bigotimes^n\!\!\cA_{02\cup 35}}H\big)$
visibly extend to an action of the von Neumann algebra $\cA(\bigsqcup^n[0,5])=\cA(I)$.
They therefore also extend to an action of $\cA(I)$ on~$H$.
\end{proof}

\section{An application}

In our recent preprint \cite{CS(pt)}, we introduced higher categorical analogs of von Neumann algebras called \emph{bicommutant categories}.
A bicommutant category is a tensor category which is equivalent to its bicommutant inside $\mathrm{Bim}(R)$.
(The latter is the category of all bimodules over a hyperfinite factor; it plays the role of the algebra of bounded operators on a Hilbert space.)
A bicommutant category is also equipped with a higher categorical analog of a $*$-structure, called a bi-involutive structure \cite[Def.\,2.3]{BicommutantFusion}. 

In \cite{CS(pt)}, we made the following announcement:
for $G$ the group $SU(n)$ and for $k$ a positive integer,
the category of positive energy representations of the based loop group of $G$ at level $k$ is a bicommutant category. Moreover, its Drinfel'd center is the category of positive energy representations of the free loop group of~$G$:\medskip
\begin{equation}\label{eq: Z(Rep(OmG))=Rep(LG)}
Z\big(\Repposen(\Omega G)\big)\,=\,\Repposen(LG).
\medskip
\end{equation}
We then argued that this result provides good evidence for our claim that the tensor category of positive energy representation of the based loop group is the value of Chern-Simons theory on a point.
\begin{remark}
The tensor category of positive energy representations of $LG$, as defined using conformal nets (see, e.g., \cite{GabbianiFrohlich, Wassermann}), has, to our knowledge, not been compared to the corresponding tensor category defined using affine Lie algebras (or vertex algebras, or quantum groups).
The right hand side of \eqref{eq: Z(Rep(OmG))=Rep(LG)} refers to the tensor category defined in \cite{Wassermann}.
\end{remark}

\begin{remark}
We expect the relation \eqref{eq: Z(Rep(OmG))=Rep(LG)} to hold true for every compact connected Lie group $G$ and every level $k\in H^4_+(BG,\mathbb Z)$.
It is conjectured by many people that all chiral WZW conformal nets have finite index (see \cite{GabbianiFrohlich}\cite{ToledanoPhD}\cite{Wassermann}\cite[\S 4.C]{CN1}\cite[\S8]{WZW}
for the definition of these conformal nets in various degrees of generality).
The finite index property is known when $G=SU(n)$ \cite{Xu2000, Wassermann}, and in a few other cases.
Our proof of \eqref{eq: Z(Rep(OmG))=Rep(LG)} relies crucially on the fact that the chiral WZW conformal nets associated to $G$ have finite index.
However, our dependence on this result is the only place where we use that $G$ is the group $SU(n)$.
\end{remark}

We can generalize  \eqref{eq: Z(Rep(OmG))=Rep(LG)} to arbitrary conformal nets of finite index.
The analog of the tensor category of positive energy representations of $\Omega G$
is the tensor category $\cT_\cA$ of \emph{solitons} of the conformal net $\cA$ \cite[\S3.0.1]{LongoXuDichotomy}.
The claim is then that the Drinfel'd center of the tensor category of solitons of $\cA$ is the braided tensor category of representations~of~$\cA$.

By definition, a soliton is a Hilbert space equipped with compatible actions of the algebras $\cA(I)$, where $I$ ranges over all subintervals of the standard circle whose interior does not contain the base point $1\in S^1$.
Equivalently, it ranges over all subintervals $I\subsetneq S^1_{\mathrm{cut}}$,
where $S^1_{\mathrm{cut}}$ is the manifold obtained from the standard circle by removing its base point and replacing it by two points:
\[
S^1:\,\,\, \begin{matrix}\tikz{\draw circle (.6);\fill (.6,0) circle (.022);}\end{matrix}
\qquad\qquad
S^1_{\mathrm{cut}}:\,\,\, \begin{matrix} \tikz{\draw (0,0) arc(5:355:.6) coordinate (x);\fill (0,0) circle (.02) (x) circle (.02);}\end{matrix}
\]

The monoidal structure on $\cT_\cA$ is defined as follows.
Let $H$ and $K$ be two solitons. 
Let $I_+$ be the upper half of $S^1$, and let $I_-$ be its lower half.
Precomposing the left action of $\cA(I_-)$ on $H$ by the map
\[
\cA(\;\!z\mapsto\bar z\;\!:I_+\to I_-)\;:\,\,\cA(I_+)^{op}\to \cA(I_-)
\]
yields a right action of $\cA(I_+)$ on $H$.
We let
\[
H\boxtimes K:=H\boxtimes_{\cA(I_+)} K.
\]
Here, $\cA(I_+)$ acts on $K$ in the usual way, and acts on $H$ on the right via the action described above.
The left actions of $\cA(I_+)$ on $H$ and of $\cA(I_-)$ on $K$ induce corresponding actions on $H\boxtimes K$.
Given an interval $J\subset S^1$, $1\not\in J$,  $-1\in \mathring J$, 
then, by the same argument as in \cite[Def.\,1.31]{CN1}, the actions of $\cA(J\cap I_+)$ and $\cA(J\cap I_-)$ on $H\boxtimes K$ extend to an action of $\cA(J)$.
By Lemma~\ref{Lem.1.9 of CN1}, it follows that 
for \emph{every} interval $J\subsetneq S^1_{\mathrm{cut}}$, the actions of $\cA(J\cap I_+)$ and $\cA(J\cap I_-)$ on $H\boxtimes K$ extend to an action
\[
\rho_J:\cA(J)\to B(H\boxtimes K).
\]
All together, these actions equip $H\boxtimes K$ with the structure of a soliton.

Given a soliton $H$, with actions $\rho_I:\cA(I)\to B(H)$ for $I\subsetneq S^1_{\mathrm{cut}}$,
its \emph{conjugate} $\hspace{.2mm}\overline {\hspace{-.2mm}H}$ is the complex conjugate Hilbert space equipped with the actions
\[
\cA(I)\xrightarrow{\cA(z\mapsto\bar z)}\cA(\bar I)^{op}\xrightarrow{\,\,\,*\,\,\,}\overline{\cA(\bar I)}\xrightarrow{\,\,\,\overline{\rho_{\bar I}}\,\,\,}\overline{B(H)}=B(\overline H).
\]
Here, $\bar I$ denotes the image of $I\subset S^1$ under the complex conjugation map $S^1\to S^1$.
The conjugation on $\cT_\cA$ squares to the identity and satisfies $\overline{H\boxtimes K}\,\cong\, \overline K\boxtimes \overline H$. 

\begin{defn}[{\cite[Def.\,5.3]{BicommutantFusion}}]
An object $\Omega$ of a tensor category $(\cT,\otimes)$ is called absorbing if it is non-zero and satisfies
\[
\qquad\qquad(X\not= 0) \,\,\Rightarrow\,\, (X\otimes \Omega \,\cong\,\Omega\,\cong\,\Omega\otimes X)\qquad\quad\forall X\in\cT.
\]
\end{defn}

\begin{remark*} If $\cT$ admits a conjugation, it is a little bit easier to check that an object is absorbing.
$\Omega\in\cT$ is absorbing if it is non-zero and if $X\otimes \Omega \,\cong\,\Omega$ for every $X\not=0$;
see the comments after \cite[Def.\,5.3]{BicommutantFusion} for a proof.
\end{remark*}

The next result, about the existence of absorbing objects, is a key ingredient in the proof, announced in \cite{CS(pt)}, that the category of solitons of a conformal net with finite index is a bicommutant category,
and that its Drinfel'd center is the category of representations of $\cA$.
The proof relies on Lemma~\ref{the lem} (which is essentially equivalent to Theorem \ref{thm: easy thm}).

Given a non-trivial conformal net $\cA$, let
\[
\Omega_\cA:=L^2(\cA(S^1_{\mathrm{cut}}))\in\cT_\cA, 
\smallskip
\]
with actions of $\cA(I)$ for $I\subsetneq S^1_{\mathrm{cut}}$ provided by the obvious inclusion $\cA(I)\to \cA(S^1_{\mathrm{cut}})$ followed by the left action of $\cA(S^1_{\mathrm{cut}})$ on its $L^2$-space.

Alternatively, the soliton $\Omega_\cA$ can be described as follows. Let its underlying Hilbert space be the vacuum Hilbert space $H_0$ of the conformal net $\cA$.
Given an interval $I=[e^{ia},e^{ib}]\subset S^1$ with $0\le a< b\le 2\pi$, let $\sqrt I:=[e^{ia/2},e^{ib/2}]$.
The square root function induces an isomorphism $\cA(I)\to\cA(\sqrt I)$.
For $I\subset S^1$, $1\not\in\mathring I$, the action of $\cA(I)$ on $H_0$ is the composite
$\cA(I)\to\cA(\sqrt I)\to B(H_0)$
of the above isomorphism with the standard action of $\cA(\sqrt I)$ on $H_0$.

The equivalence between the above two descriptions of $\Omega_\cA$ is provided by the linear map
$L^2(\cA(\tikz{\useasboundingbox (-.22,-.12) rectangle (.28,.25); \node[scale=1.1]{$\scriptstyle \sqrt{\;\!\cdot\;\!}$};})):L^2(\cA(S^1_{\mathrm{cut}}))\to L^2(\cA([e^{i0},e^{i\pi}]))=H_0$.

\begin{thm}\label{THM3}
Let $\cA$ be a non-trivial conformal net.
Then $\Omega_\cA\in\cT_\cA$ is characterized up to isomorphism by the following three properties:
\renewcommand{\labelenumi}{(\hspace{-.25mm}{\em\alph{enumi}}\hspace{.2mm})}
\begin{enumerate}
\item
it is non-zero,
\item
$\Omega_\cA\oplus \Omega_\cA\cong \Omega_\cA$, and
\item
the actions of $\cA(I)$ for $I\subset S^1$, $1\not\in\mathring I$, factor through an action of $\cA(S^1_{\mathrm{cut}})$.
\end{enumerate}
Moreover, if $\cA$ has finite index, then $\Omega_\cA$ is an absorbing object.
\end{thm}

\begin{proof} 
We first check that $\Omega_\cA$ satisfies the above three properties.
The first one is obvious. The third one holds by construction.
The second property is a consequence of Lemma \ref{lem: type not II_1} below:
since $\cA(S^1_{\mathrm{cut}})$ is an infinite factor,
the Hilbert spaces $L^2(\cA(S^1_{\mathrm{cut}}))\oplus L^2(\cA(S^1_{\mathrm{cut}}))$ and $L^2(\cA(S^1_{\mathrm{cut}}))$ are isomorphic as left $\cA(S^1_{\mathrm{cut}})$-modules.
It follows that $\Omega_\cA\oplus \Omega_\cA\cong \Omega_\cA$.

Let $\Omega'$ be another soliton that satisfies the same three properties.
Then $\Omega'$ is naturally an $\cA(S^1_{\mathrm{cut}})$-module.
By the classification of modules over factors, there is a unique non-zero $\cA(S^1_{\mathrm{cut}})$-module (up to isomorphism) that satisfies $\Omega'\oplus \Omega'\cong \Omega'$.
It follows that $\Omega'\,\cong\, \Omega_\cA$.

Let us now assume that $\cA$ has finite index.
Given a non-zero soliton $X$, we need to show that $X\boxtimes \Omega_\cA\cong \Omega_\cA$.
Equivalently, we need to show that $X\boxtimes \Omega_\cA$ satisfies the three properties listed above.
The first property, $X\boxtimes \Omega_\cA\not =0$, holds because fusing over a factor sends non-zero Hilbert spaces to non-zero Hilbert spaces (see, e.g., \cite[Prop.\,5.2]{dualizability}).
The second property, $X\boxtimes \Omega_\cA\oplus X\boxtimes \Omega_\cA\cong X\boxtimes \Omega_\cA$, is an immediate consequence of the corresponding property of $\Omega_\cA$.
The third property is more tricky and its verification will occupy the rest of this proof.

Let $\cI_0$ be the collection of all subintervals of $S^1$ whose interior does not contain the base point $1\in S^1$.
Equivalently, $\cI_0$ is the collection of all subintervals $I\subsetneq S^1_{\mathrm{cut}}$.
By definition, a soliton is a Hilbert space equipped with compatible actions of all the algebras $\cA(I)$ for $I\in\cI_0$. 
Let $I_1:=[e^{i0},e^{i\pi/2}]$ and $I_4:=[e^{i3\pi/2},e^{i2\pi}]$ be the first and fourth quadrants of the standard circle, and let $I_{14}:=I_1\sqcup I_4\subset S^1_{\mathrm{cut}}$ be their disjoint union
(whereas $I_1$ and $I_4$ are not disjoint in $S^1$, these intervals \emph{are} disjoint when viewed as subsets of $S^1_{\mathrm{cut}}$).
The collection $\cI_0$ is not a $2$-cover of $S^1_{\mathrm{cut}}$, because no element of $\cI_0$ contains $\partial S^1_{\mathrm{cut}}$.
But
\[
\cI:=\cI_0\cup\{I_{14}\}
\]
is a $2$-cover.

Recall that, by definition, $X\boxtimes \Omega_\cA=X\boxtimes_{\cA(I_+)} \Omega_\cA$.
By the split property, the actions of $\cA(I_+)$ and of $\cA(I_4)$ on $\Omega_\cA$ extend to an action of their spatial tensor product.
Equivalently, there exists an intermediate type $I$ factor:
\[
\qquad\qquad\cA(I_4)\subset N\subset \cA(I_+)'\qquad \text{\small(commutant inside $B(\Omega_\cA)$)}.
\]
The action of $N$ on $\Omega_\cA$ commutes with that of $\cA(I_+)$ and thus induces an action on $X\boxtimes_{\cA(I_+)} \Omega_\cA$.
The latter commutes with the action of $\cA(I_1)$ coming from $X$, so we get an intermediate type $I$ factor:
\[
\qquad\qquad\cA(I_4)\subset N\subset \cA(I_1)'\qquad \text{\small(commutant inside $B(X\boxtimes_{\cA(I_+)} \Omega_\cA)$)}.
\]
Equivalently, the actions of $\cA(I_1)$ and $\cA(I_4)$ on $X\boxtimes_{\cA(I_+)} \Omega_\cA$ extend to their spatial tensor product $\cA(I_{14})$.

The above action of $\cA(I_{14})$ on $X\boxtimes \Omega_\cA$, together with the actions of $\cA(I)$ for $I\in \cI_0$ coming from the fact that $X\boxtimes \Omega_\cA$ is a soliton,
assemble to a compatible family of actions
\[
\qquad\qquad\rho_I:\cA(I)\to B(X\boxtimes \Omega_\cA)\qquad\quad \forall I\in \cI.
\]
Finally, by Lemma~\ref{the lem}, since $\cI$ is a 2-cover, these extend to an action of $\cA(S^1_{\mathrm{cut}})$.
This finishes the proof of condition (\hspace{-.2mm}{\em c}\hspace{.2mm}).
\end{proof}

\section{Appendix}

For classical conformal nets, it is well known that, unless $\cA(I)=\mathbb C$ for all $I$, the algebras $\cA(I)$ are hyperfinite $\mathit{III}_1$ factors \cite[Thm.\,2.13]{GabbianiFrohlich}.
Hyperfiniteness is a formal consequence of the split property, and holds equally well in the coordinate free setup (i.e., for conformal net as in \cite[Def.\,1.1]{CN1}).
Indeed, given an interval $I$, write $\mathring I=\bigcup I_n$ with $I_n\subset \mathring I_{n+1}$.
By the split property, there exist intermediate type $I$ subfactors $\cA(I_n)\subset N_n \subset \cA(I_{n+1})$, and so $\cA(I)=\bigvee N_n$ is hyperfinite.

We do not know how to prove the type $\mathit{III}_1$ property in the coordinate free setup.
The following lemma is the best we can offer:

\begin{lem}\label{lem: type not II_1}
Let $\cA$ be a non-trivial conformal net.
Then the algebras $\cA(I)$ are infinite factors (they are infinite dimensional, and they are not of type $\mathit{II}_1$).
\end{lem}
\begin{proof}
The algebra $\cA(I)$ is infinite dimensional as it contains infinitely many non-trivial commuting subalgebras. 

Let $I_0$ be the upper half of the standard circle, so that the vacuum sector $H_0$ is $L^2(\cA(I_0))$.
Assume by contradiction that the algebra $\cA(I_0)$ is of type $\mathit{II}_1$.
Then the von Neumann dimension of $H_0$ as an $\cA(I_0)$-module is equal to $1$.
By diffeomorphism covariance, for every interval $I\subset S^1$, the dimension of $H_0$ as an $\cA(I)$-module is also $1$.
Given two intervals $I\subsetneq J\subset S^1$, we have
\[
\dim_{\cA(I)}(H_0)\,=\,[\cA(J):\cA(I)]\cdot\dim_{\cA(J)}(H_0).
\]
It follows that $[\cA(J):\cA(I)]=1$. The inclusion $\cA(I)\to \cA(J)$ is therefore an isomorphism, a contradiction.
\end{proof}

\section{Acknowledgements}
\noindent
This project has received funding from the European Research Council (ERC)
under the European Union's Horizon 2020 research and innovation programme (grant agreement No 674978).


\end{document}